\documentclass[10pt,conference]{IEEEtran}
\IEEEoverridecommandlockouts
\usepackage{amsmath,amssymb,amsfonts}
\usepackage{amsthm}
\newtheorem{definition}{Definition}
\newtheorem{example}{Example}
\newtheorem*{remark}{Remark}

\newtheorem{lemma}{Lemma}

\usepackage[ruled]{algorithm2e}
\usepackage{algpseudocode}
\usepackage[braket, qm]{qcircuit}
\usepackage{graphicx}
\usepackage{textcomp}
\usepackage{xcolor}
\usepackage{placeins}
\usepackage{hyperref}
\usepackage{makecell}
\usepackage{float}
\floatstyle{plaintop}
\restylefloat{table}
\usepackage[numbers,sort&compress]{natbib}

\def\BibTeX{{\rm B\kern-.05em{\sc i\kern-.025em b}\kern-.08em
    T\kern-.1667em\lower.7ex\hbox{E}\kern-.125emX}}
\begin{document}

\title{Mapping quantum circuits to modular architectures with QUBO\\}

\author{\IEEEauthorblockN{Medina Bandic \IEEEauthorrefmark{2}\IEEEauthorrefmark{1}, Luise Prielinger \IEEEauthorrefmark{2}\IEEEauthorrefmark{1},
Jonas Nüßlein \IEEEauthorrefmark{5},
Anabel Ovide \IEEEauthorrefmark{3},
Santiago Rodrigo \IEEEauthorrefmark{4},
Sergi Abadal \IEEEauthorrefmark{4}, \\
Hans van Someren \IEEEauthorrefmark{1}, Gayane Vardoyan \IEEEauthorrefmark{1},
Eduard Alarcon \IEEEauthorrefmark{4},
Carmen G. Almudever \IEEEauthorrefmark{3}
and
Sebastian Feld \IEEEauthorrefmark{1}}
\\
\IEEEauthorblockA{\textit{ \IEEEauthorrefmark{1} Delft University of Technology (QuTech)},
The Netherlands\\
}
\IEEEauthorblockA{\textit{\IEEEauthorrefmark{3} Universitat Politècnica de Valencia}, 
Spain\\
}
\IEEEauthorblockA{\textit{\IEEEauthorrefmark{4} Universitat Politècnica de Catalunya},
Spain\\
}
\IEEEauthorblockA{\textit{\IEEEauthorrefmark{5} LMU Munich},
Germany\\
}
\\
\IEEEauthorblockA{ \IEEEauthorrefmark{2} These authors contributed equally to this work.
}}

\maketitle

 \begin{abstract}
Modular quantum computing architectures are a promising alternative to monolithic QPU (Quantum Processing Unit) designs for scaling up quantum devices. They refer to a set of interconnected QPUs or cores consisting of tightly coupled quantum bits that can communicate via quantum-coherent and classical links. In multi-core architectures, it is crucial to minimize the amount of communication between cores when executing an algorithm. Therefore, mapping a quantum circuit onto a modular architecture involves finding an optimal assignment of logical qubits (qubits in the quantum circuit) to different cores with the aim to minimize the number of expensive inter-core operations while adhering to given hardware constraints. In this paper, we propose for the first time a Quadratic Unconstrained Binary Optimization (QUBO) technique to encode the problem and the solution for both qubit allocation and inter-core communication costs in binary decision variables. To this end, the quantum circuit is split into slices, and qubit assignment is formulated as a graph partitioning problem for each circuit slice. The costly inter-core communication is reduced by penalizing inter-core qubit communications. The final solution is obtained by minimizing the overall cost across all circuit slices. To evaluate the effectiveness of our approach, we conduct a detailed analysis using a representative set of benchmarks having a high number of qubits on two different multi-core architectures. Our method showed promising results and performed exceptionally well with very dense and highly-parallelized circuits that require on average 0.78 inter-core communications per two-qubit gate.  
\end{abstract}


\section{Introduction}
    
A major challenge for current NISQ (Noisy Intermediate Scale Quantum) devices is scalability. These processors suffer from high error rates and limited qubit counts and connectivity, which hinder the demonstration of the full potential of quantum computing. Well-known quantum algorithms, such as Shor's algorithm for factoring large numbers and Grover's algorithm for searching an unsorted database, are expected to provide significant speedups over classical algorithms, but these benefits will likely only be realized with quantum computers that have far more qubits than current systems, which are typically limited to hundreds of qubits. 
Most present-day quantum computers are implemented as single-processor devices, i.e., a single chip that contains all qubits. These designs are hard to scale mostly due to crosstalk \cite{sarovar2020detecting} and limitations related to control electronics \cite{lao2021timing}. 
An alternative architectural design to scale quantum computers, similar to classical computing, lies in multi-processor (multi-core) computers, which have already been proposed by various quantum processor manufacturers.~\cite{Monroe_2014, laracuente2023modeling, https://doi.org/10.48550/arxiv.2201.08861,https://doi.org/10.48550/arxiv.2210.10921}.

These new architectures will enable 
distributed multi-core quantum computing, that is, executing a large algorithm consisting of more qubits than there are in a single
processor by distributing it over different cores. 
In this case, similarly to the resource-constrained NISQ devices, a quantum circuit mapping process \cite{li2019tackling} is required to ensure the efficient use of hardware resources and in turn to maximize the algorithm success rate. Considering that qubit interactions within a core are negligible when compared to those in between cores in terms of operation time and fidelity \cite{baker2020time}, quantum circuit mapping in the multi-core regime is mainly focused on minimizing
the amount of inter-core communications (operations between qubits that are in different processors). This is mostly achieved by finding a good assignment or allocation of the logical qubits (qubits in the circuit) to the physical qubits in different cores and by optimally performing inter-core operations when required.

Finding an optimal solution for the mapping problem can be computationally infeasible for a large number of qubits, even for single-core devices \cite{Siraichi2018}. To address this challenge, various mapping algorithms have been proposed 
\cite{bandic2020structured}. However, due to the early-stage development of modular quantum computing architectures, just a few quantum circuit mapping techniques \cite{baker2020time,ferrari2020compiler} have been explored, and they are only tested on architectures with a small number of qubits and simple, unrealistic all-to-all connectivity between qubits with limited benchmark sets. 

In contrast to prior approaches, we address the circuit mapping problem for multi-processor devices using a quadratic unconstrained binary optimization (QUBO) formulation. This novel approach is employed for solving qubit allocation, as well as inter-core qubit communication and routing.  
Among the advantageous properties of this QUBO-based approach, four stand out: 
\begin{itemize}
    \item[i.] The QUBO formulation a priori encompasses the entire solution landscape for the quantum circuit mapping problem, and therefore does not exhibit limitations stemming from a reliance on look-ahead functions. The latter or other local estimates are a common strategy for decision-making regarding qubit placement and transfer \cite{baker2020time}.
    \item[ii.] Unlike previous mapping solutions, the optimization process is decoupled from the objective function, thereby enabling the selection of a tailored optimization method, depending on the size and scope of the problem at hand. More precisely, the method can be flexibly adapted to a use case by choosing a suited exact or heuristic solver without reformulating the method, for a small or large quantum circuit to be mapped, respectively.
    \item[iii.] The aforementioned separation between the objective function and finding its optimal solutions allows for the latter to evolve with the development of new solving approaches. This development includes not only even more powerful software-based solvers \cite{punnen2022quadratic} for QUBO instances, but also the construction of single-purpose quantum hardware, a so-called quantum annealer \cite{dwaveweb}. 
    \item[iv.] The k-partitioning problem \cite{chopra1993partition}, which is often used in qubit mapping for multi-core systems \cite{baker2020time}, can be expressed as a quadratic objective function and has been well-defined in the literature \cite{ushijima2017graph}. This makes it an ideal candidate for a QUBO formulation. 
\end{itemize}

The aim of our method is to find the optimal placement of qubits for all circuit slices (Fig. \ref{fig:objective}) to a multi-core system by representing the per-slice qubit allocation as a k-partitioning problem and penalizing the cost of any inter-core communication simultaneously.   
This method is tested by mapping an extensive set of benchmark circuits, that cover a substantial range of parameters, including the number of qubits, circuit depth, gate density and structure, to two different multi-core architecture topologies. We hope that this work represents a stepping stone for the development of new mapping techniques for modular architectures that will be essential components of future full-stack quantum systems \cite{bandic2022full}. 

The paper is organized as follows: Sec. \ref{Sec2} introduces the quantum circuit mapping problem, discusses the challenges of multi-core quantum computation, and presents the concept of QUBO in general. In Sec. \ref{Sec3} our QUBO-based approach for mapping quantum circuits to multi-core devices is described in detail, including the method's objectives, definition and proofs of the objective function. Sec. \ref{Sec4} presents the experimental setup, the benchmark set used and the chosen hardware platforms. Sec. \ref{Sec5} discusses the obtained results.
Finally in Sec. \ref{Sec6} we describe possible future directions including potential solutions to overcome current limitations and conclude the paper.

\section{Background and related work}

\label{Sec2}

\subsection{Mapping of quantum circuits: from single to multiple cores}

\begin{figure}[h!]
\centering
    \includegraphics[width=0.7\linewidth]{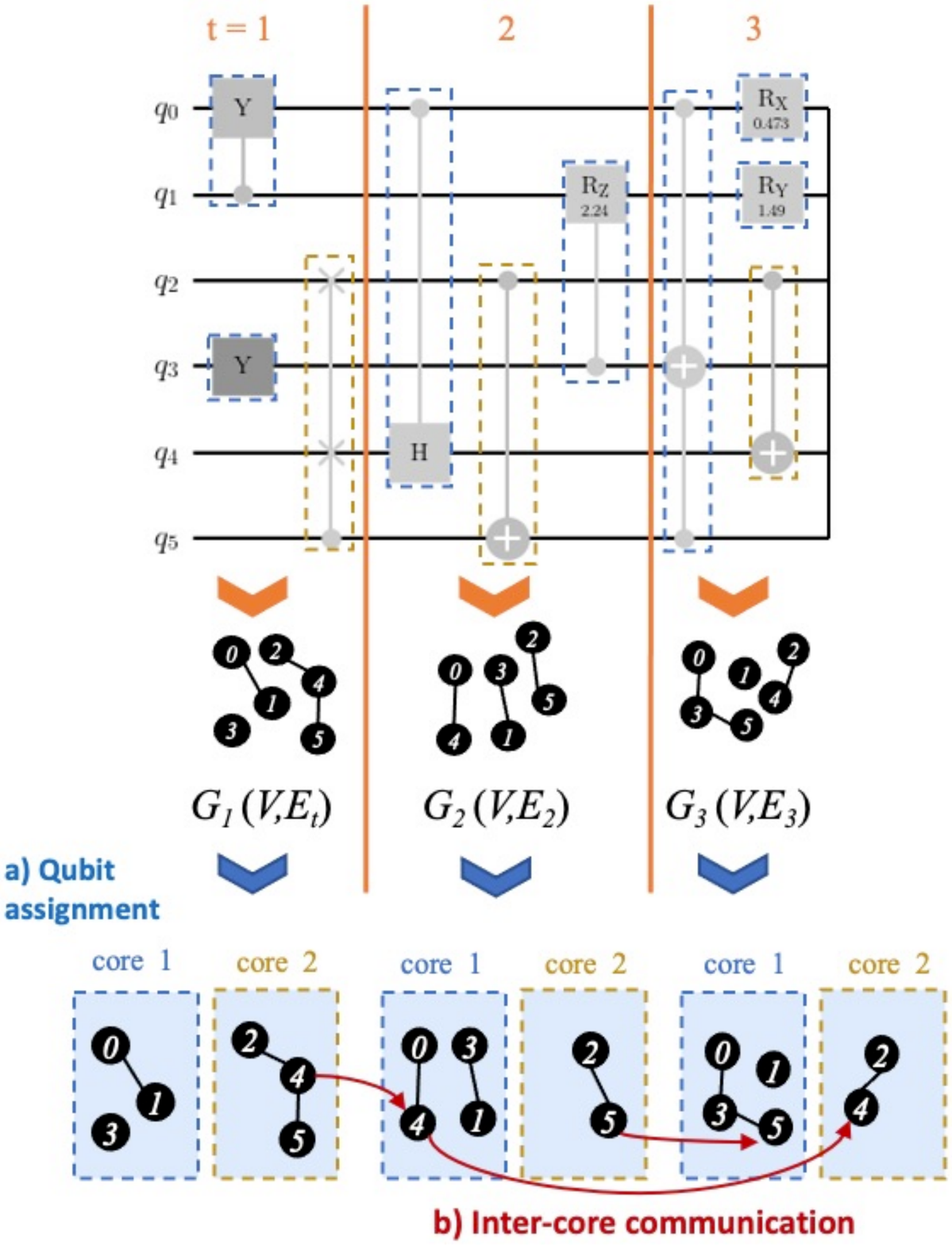}
    \caption{Graphical outline of Ex. \ref{ex:mapping}. A quantum circuit is split into three circuit slices $t = \{1,2,3\}$ (where the circuit slices' borders are marked by solid vertical lines in orange) represented by the respective interaction graphs $G_t(V,E_t)$. The partitioned circuit is mapped to a $k=2$ core system (two blue shaded rectangles below each circuit slice), each with a capacity of $c_1=c_2=4$ physical qubits. \textbf{a)} Logical qubits (black circles) sharing a gate interaction are assigned to the same core (dashed rectangles). The first gate $(0, 1)$ of the circuit is assigned to core 1 marked in blue dashed rectangles, while the three-qubit gate involving qubits (2,4,5) is assigned to core 2.\textbf{ b)} If a qubit is assigned to a different core in subsequent slices, it must undergo a state transfer, i.e., an inter-core communication, shown by red arrows. Hence, the solution shown requires three inter-core qubit communications.}
    \label{fig:objective}
\end{figure}

Due to the previously mentioned limitations of current NISQ devices, algorithms (in the form of quantum circuits) may need to be modified in order to be executed, through a process called quantum circuit mapping \cite{bandic2020structured}. This process involves: i) assigning the logical qubits of the quantum circuit to the physical qubits of the quantum chip, an example of which is provided in Fig. \ref{fig:objective}; ii) moving (routing) the logical qubits that need to interact in form of two-qubit gates to adjacent positions on the chip. Note that this is necessary as physical qubits for most quantum technologies are not all-to-all connected; and in some cases includes iii) scheduling operations to maximally parallelize them with the aim of lowering the execution time, which is important due to the decoherence of qubits \cite{bandic2022full}. The quantum circuit mapping process can vary from device to device due to differences in technology and qubit connectivity. It is essential to perform this task to guarantee the effective utilization of scarce hardware resources and to decrease the likelihood of errors that may arise during the execution of quantum operations by minimizing the number of additional gates and the circuit latency. However, solving the quantum circuit mapping problem for a large number of qubits can be computationally unfeasible, even for contemporary single-core devices. To tackle it, diverse mapping algorithms have been proposed, including heuristic or brute-force strategies, graph-theoretical techniques, dynamic programming algorithms, and machine learning-based solutions \cite{zulehner2018efficient,wille2016look, li2019tackling, lao2021timing, itoko2020optimization,pozzi2020using,jiang2021quantum,li2020qubit,tannu2019not, steinberg_topological-graph_2022, bandic2020structured}.

These mapping techniques used in single-processor NISQ devices are however not applicable to multi-core quantum computing architectures, which have emerged as a promising approach to scale-up quantum computing systems. 
In the multi-core architectural approach, cores (QPUs) are connected with classical and quantum communication links. Quantum links allow to 
`move' quantum states from one core to another or to perform inter-core quantum gates \cite{10.1145/3579367} (e.g., remote extended CNOT) making use of entanglement.  
Classical links are required to assist core coordination and job distribution \cite{Rodrigo2021}. The architectural complexity involving the communication channels and traffic make it more difficult to perform quantum circuit mapping when compared with mapping to single-core devices \cite{rodrigo2022characterizing, ovide2023mapping}.

This consequently led to the development of new techniques for multi-core architectures whose  main aim is to reduce the number of expensive inter-core operations. To achieve this, an optimal assignment 
of the logical qubits to the physical qubits of different cores is crucial.  
The literature on this topic is limited: some proposals focused on the quantum compilation for distributed quantum computing \cite{cuomo2023optimized,ferrari2020compiler}, whereas others focused only on  the quantum circuit mapping part of compilation \cite{baker2020time}, which is closer to our work. In \cite{baker2020time}, the quantum circuit is split into smaller partitions. Interaction graphs representing operations within a circuit slice are then mapped onto cores by grouping the qubits with the highest amount of interactions together, and while taking into account different look-ahead approaches. This was one of the  first proposals for mapping quantum algorithms to multi-core devices, but it was only tested on a very simple all-to-all architecture and with a limited set of benchmarks. 

In order to improve their initial approach in terms of circuit and graph partitioning with an overly-restricting local 
look-ahead function, we rely on previous subgraph isomorphism- \cite{sivarajah_tketrangle_2021, Siraichi19,steinberg_topological-graph_2022,li_qubit_2021,jiang2021quantum} and QUBO optimization-based single-core solutions \cite{Dury2020,prielinger2023aquadratic} where now instead of mapping logical qubits to optimal subsets of coupling graph representing the physical qubits and their connections (Fig. \ref{fig:objective}) we map qubits to different cores. In our approach, we use a k-partitioning-based QUBO formulation, which is, together with the mapping procedure explained in more detail in the following sections.

\subsection{Quadratic unconstrained binary optimization}
\label{Sec2.2}
The introduction of Quadratic Unconstrained Binary Optimization, termed QUBO or UBQP in literature, goes back to the 1960s \cite{hammer1969pseudo}. We will use the following definition
\begin{definition}[Quadratic unconstrained binary optimization] 
    Quadratic Unconstrained Binary Optimization (QUBO) formulation presents a NP-hard \cite{punnen2022quadratic} mathematical problem to be optimized. The objective function to be minimized is expressed by the formula:
    {\footnotesize
    \begin{equation}
        \underset{\mathbf{x}}{\min}\ \mathbf{x}^T Q \mathbf{x} = \underset{\mathbf{x}}{\min} \sum_{i<j} Q_{ij} x_i x_j + \sum_i Q_{ii} x_i, 
    \end{equation}}
    where $\mathbf{x}$ is a vector of binary decision variables of size $N$, i.e., $\mathbf{x} = \{0,1\}^N$ and $Q$ is a symmetric, square matrix of $N\times N$ real valued constants \cite{punnen2022quadratic}.
\end{definition}
In recent years the QUBO model gained scientific attention which led to the development of a wide range of applications in combinatorial optimization, as well as of effective solver techniques \cite{glover2018tutorial, lucas2014ising, Calude2017, Neukart2017, Date2021, nusslein2022black}. Unlike many other optimization formulations, a QUBO instance can be solved with a quantum annealer, as for example the one constructed by the hardware provider D-Wave Systems Inc.\cite{dwaveweb}. Even though the computational solving ability of these devices is still restricted to small problem sizes, the rapid evolution of these systems provides hope that more realistic problems can be addressed in the near to mid-term future \cite{Yarkoni2022, vanDam2001}. 
Therefore QUBO instances are in practice still solved with either exact or approximate classical procedures on classical computers. Unlike heuristics or approximation algorithms, exact QUBO solvers guarantee that the solution found is optimal. However, the NP-hard nature of a QUBO problem implies that exact solvers result in infeasible running times for large problem sizes \cite{kochenberger2014unconstrained}. Therefore, approximate or heuristic methods are often used to find near-optimal solutions within reasonable time limits. Prominent examples are simulated annealing, tabu search and steepest descent, but also commercial cloud-based  solvers \cite{oshiyama2022benchmark}. A comprehensive literature review of QUBO applications and solving methods can be found in \cite{kochenberger2014unconstrained, dunning2018works, punnen2022quadratic}.

\section{Solving the mapping problem for distributed multi-core quantum computing with QUBO}
\label{Sec3}
In this section, we will introduce our mapping approach and explain its objectives in detail including an illustrative example. We will then give the required proofs and see how solutions can differ due to a simple scalar weighting factor $\lambda$, which we use to scale the different components of the objective function.

\subsection{Objectives of the QUBO mapping}

The quantum circuit to be mapped is split into circuit slices, each consisting of a sequence of gates. This partition is generated via a recursive slicing procedure, for which we will provide more details in Sec. \ref{Slicing}. A circuit slice comprises a set of quantum gates that can be performed without the use of inter-core communication\footnote{The definition of a circuit slice is specific to this study and may not match the concept of a time slice in other publications.}. A circuit slice $t$, therefore, presents an interaction graph $G_t(V,E_t)$, in which the logical qubits of the circuit form the node set $V$ and set specific qubit interactions\footnote{Only qubit interactions with $n>1$ are relevant to the mapping procedure since single-qubit gates are not core-dependent.} form edge set $E_t$. Each edge $e\in E_t$ denotes at least one gate operation between qubits present in circuit slice $t$. The objective function should then a) find an assignment for each time-slice graph, i.e., all qubits involved in the same gate should be assigned to the same core without exceeding the core's capacity $c_j$ the maximum number of logical qubits a core $j$ can hold; and b) minimize inter-core communications between assignments, where the latter refers to the relocation of a logical qubit state between two different cores of a distributed quantum system. Tab. \ref{tab:nota} summarizes the notation used in this study.

\begin{table}[ht]
    \centering
    \resizebox{0.8\linewidth}{!}{
    \begin{tabular}{r|l}
        Symbol & Description \\ \hline
        $n$ & \text{Number of logical qubits in a quantum circuit} \\
        $k$ & Number of cores in the multi-core system\\
        $c_j$ & \makecell[l]{Capacity, number of qubits that a core $j$ can hold}\\
        $T$ & Total number of circuit slices\\
        \hline \vspace{0.2cm}
    \end{tabular}}
    \caption{Notation used in this study.}
    \label{tab:nota}
\end{table}

Let us consider the following example:
\begin{example}[Qubit assignment and inter-core communication] \label{ex:mapping}
   Given a six-qubit circuit with seven multi-qubit\footnote{The method is in general not restricted to how many qubits the gates of the circuit involved. Note, however, that most quantum processors only support up to two-qubit gates.} gates depicted in Fig. \ref{fig:objective} split into circuit slices $t=\{1,2,3\}$ (details on slicing are given in  \ref{Slicing}). The objective is to map the former to a two-core system, each core with a capacity of $c_1=c_2=4$. We can write the circuit interactions as follows:
   \begin{align}
        \begin{array}{c|c}
        t & Qubits \\ \hline
        1 & (0,1), (2,4,5), 3\\
        2 & (0,4), (2,5), (1,3)\\
        3 & (0,3,5), (2,4), 1\\
        \end{array} 
   \end{align}
   where each row holds the six logical qubits, and tuples hold qubits involved in the same multi-qubit gate. In each circuit slice, all $n$ qubits need to be assigned to cores, such that qubits sharing a tuple also share the same core $j$, and the core`s capacity $c_j$ is not exceeded. Possible assignments are
   \begin{align}
        \begin{array}{c|c}
        t  & Assignment \\ \hline
        1 & {(0,1),3 \mapsto core_1}; {(2,4,5) \mapsto core_2}\\
        2 &  {(0,\textcolor{red}{4}), (1,3) \mapsto core_1}; {(2,5) \mapsto core_2}\\
        3 & {(0,3,\textcolor{red}{5}),1\mapsto core_1}; {(2,\textcolor{red}{4}) \mapsto core_2}\\
        \end{array} 
   \end{align}
    where figures in red indicate inter-core quantum state transfers of qubits $4$ and $5$. Namely,
   \begin{align}
        \begin{array}{c|c|c}
        t  & \text{Qubit}\ 4 & \text{Qubit}\ 5 \\ \hline
        1-2  & core_2 \rightarrow core_1 &\\
        2-3 &  core_1 \rightarrow core_2  & core_2 \rightarrow core_1 \\
        \end{array}.
   \end{align}

\end{example}
The assignment portion of the problem can be cast into an instance of the $k$-partitioning problem, a well-known graph theoretical problem \cite{chopra1993partition} where $k$ denotes the number of subsets in the partition.
We accomplish this as follows:
\begin{definition}[Qubit Assignment] \label{def:qubitassignment}
    Let $G_t=(V,E_t)$ be an interaction graph of a circuit slice $t$ with logical qubit set $V$ and edge set $E_t$, where an edge $e\in E_t$ denotes at least one gate operation between qubits present in circuit slice $t$. For fixed integers $k$ and $c_j$, the problem is to find a partition $\Phi^t = (\Phi^t_1, \dots, \Phi^t_k) $, of the qubit set $V$ into at most $k$ subsets, such that the subsets $|\Phi^t_j| \le c_j$, and the number of cut edges is minimized. A cut edge is defined as an edge whose endpoints are in different subsets $\Phi_j, \Phi_l$, where $j\ne l$.
\end{definition}

\subsection{The objective function}
In order to define the objective function, we first need to introduce binary decision variables:
\begin{definition}[Binary Solution Array $\mathbf{x}$] \label{def:solarray}
A solution vector $\mathbf{x}$ comprises $T$ circuit slices $\mathbf{x} = [\mathbf{x}^1, ..., \mathbf{x}^T]$, where a solution $\mathbf{x}^t$ for time slice $t$ presents a qubit assignment. The latter is again comprised of $k$ subsets
{\footnotesize
\begin{equation}
\mathbf{x}^t = [\mathbf{x}_{1}^t,... \mathbf{x}_{k}^t]
\end{equation}}%
where a vector $\mathbf{x}_{j}^t$ holds $n=|V|$ solution variables. A solution variable $x^t_{ij} = 1$ indicates that qubit $i \in V$ is assigned to core $j$ in circuit slice $t$ and $x^t_{ij} =0$ if it is not assigned to the core $j$. This yields a problem size of $N=T\cdot k \cdot n$, where $T,k$ and $n$ are the numbers of circuit slices, cores and logical qubits, respectively.
\end{definition}

In order to impose core capacities $c_j$ during the mapping of logical qubits, we will make use of so-called slack variables. Using slack variables is a well-established technique for replacing inequalities with equations in mathematical programming problems \cite{murtagh1978large}.
Let us consider the following example:
\begin{example}
    Assume we are given a circuit slice $t$ of a three-qubit circuit and that core $j$ has capacity $c_j=2$. The inequality constraint 
    {\footnotesize
    \begin{equation}
        \sum^3_{i=1} x^t_{ij} \le 2 \label{eq:ex:slack1}
    \end{equation}}%
    states that we can assign no more than two qubits to core $j$. In order to write the inequality constraint as an equation, which we can then use as a quadratic objective function term, we introduce two slack variables, $y^t_{i'j}\in \{0,1\}^2$, where $y^t_{i'j}=1$ signifying a physical qubit $i'$ being occupied by a logical qubit and $y^t_{i'j}=0$ otherwise. We can then write
    {\footnotesize
    \begin{equation}
        \sum^3_{i=1} x^t_{ij} - (y^t_{1j} + y^t_{2j}) = 0.\label{eq:ex:slack2}
    \end{equation}}
    Eqs. (\ref{eq:ex:slack1}) and (\ref{eq:ex:slack2}) are equivalent. One can always ensure compliance of the solution variables to (\ref{eq:ex:slack1}) by choosing suitable slack variables in (\ref{eq:ex:slack2}). As an example, if $x^t_{1j}=x^t_{2j}=0,\ x^t_{3j}=1$, then one of the slack variables is set to one, e.g., $y^t_{1j}=1,\  y^t_{2j}=0$. If Eq. (\ref{eq:ex:slack1}) is violated, however,  i.e., $x^t_{1j}=x^t_{2j}=x^t_{3j}=1$, Eq. (\ref{eq:ex:slack2}) cannot be satisfied either, since the sum of the slack variables cannot exceed two.
\end{example}

With Defs. \ref{def:qubitassignment} and \ref{def:solarray}, we can set up the first part of the quadratic binary objective function, namely the assignment part. To do so, we adapt the k-partitioning instance given by \cite{ushijima2017graph} for a circuit slice $t$, termed $F(\mathbf{x}^t)$. The latter encodes the optimal solution(s) $\mathbf{x}^{t*}$ at its minimum.
{\footnotesize
\begin{align}
    \min_{\mathbf{x}^t} \ &F(\mathbf{x}^t) = \min_{\mathbf{x}^t} \ S(\mathbf{x}^t) + P(\mathbf{x}^t) + R(\mathbf{x}^t) \label{eq:F}
\end{align}}%
with
{\footnotesize
\begin{align}
S(\mathbf{x}^t) &:= \sum^n_{i=1}\left(\sum^k_{j=1} x^t_{ij} - 1\right)^2 \label{eq:S},\\ 
            R(\mathbf{x}^t) &:= \sum^k_{j=1}\left(\sum^n_{i=1} x^t_{ij} - \sum^{c_j}_{i'=1} y^t_{i'j}\right)^2\label{eq:R},\\
            P(\mathbf{x}^t) &:= \sum^k_{j=1} \mathbf{x}^{tT}_{j} L_t  \mathbf{x}^t_{j} \label{eq:P},
\end{align}}%
where $\mathbf{x}_j^t\in \{0,1\}^{n}$ is the binary solution vector of a circuit slice $t$ belonging to a core $j$ and the superscript $T$ in $\mathbf{x}^{tT}$ denotes the transposed version of the vector $\mathbf{x}^{t}$. A minimum value of $S$ ensures that a qubit $i$ is assigned to exactly one core, $R$ penalizes an assignment that exceeds the
core’s capacity with $c_j$ additional binary slack variables $\textbf{y}^t_j = \{0,1\}^{c_j}$ for each core $j$, and lastly, $P$ serves to penalize any cut edge between cores using the graph Laplacian $L_t$ of circuit slice $t$. Given an undirected interaction graph $G_t = (V, E_t)$, the graph Laplacian $L_t$ is the $n\times n$ matrix 
{\footnotesize
    \begin{align}
        L_t = D_t - A_t
    \end{align}}%
    where $D_t$ is the degree matrix and $A_t$ is the adjacency matrix of the graph $G_t$  \cite{weisstein_laplacian_nodate}.

    \begin{lemma}
        If $\mathbf{x}_{opt}^t$ is a binary solution vector for which 
        $F(\mathbf{x}_{opt}^t)=0$ then $\mathbf{x}_{opt}^t$ represents an assignment of all logical qubits with the following properties: each qubit is mapped to exactly one core $j$; each adjacent pair of qubits in the graph $G_t(V,E_t)$ is assigned to the same core; and $\forall j\in\{1,\dots,k\}$ no more than $c_j$ qubits are assigned to core $j$. We refer to $\mathbf{x}_{opt}^t$ as a valid assignment.
    \end{lemma}

    \begin{proof} 
    Since $S$ only consists of sums comprising quadratic terms with a minimum value of zero, $S$ is zero, iff for each $i\in V$ exactly one variable $\{x^t_{ij}|1 \le j \le k\}$ has value $1$. We call this property \textit{s}.
    
    $R$ is zero iff for each $1 \le j \le k$ 
    {\footnotesize
     \begin{align*}
         \sum^n_{i=0} x^t_{ij} = \sum^{c_j}_{i'=0} x^t_{ij}.
     \end{align*}}
    This can only hold, if $c_j$ or fewer variables of $\{x^t_{ij}|i\in V\}$ have a value of $1$. Otherwise $\sum^n_{i=0} x^t_{ij} > c_j$ and $R$ exhibits a value greater. We call this property \textit{r}.

    $P$ is zero iff for each $1 \le j \le k$ 
     {\footnotesize
     \begin{align*}
         \mathbf{x}^{T}_{j} L \mathbf{x}_{j}= \sum^n_{i=1} d_i x_{ij}^2- \sum^n_{i=1} \sum_{v\in V_i}x_{ij}x_{vj} = 0.
     \end{align*}}
     where $V_i$ denotes the set of adjacent vertices of vertex $i$. The above can only hold if for each $i$, $\sum_{v\in V_i}x_{ij}x_{vj} = |V_i| = d_i$. Otherwise, if there is less than $d_j$ adjacent vertices assigned to core $j$, $\sum_{v\in V_i}x_{ij}x_{vj} < |V_i|$ and $\mathbf{x}^{T}_{j} L \mathbf{x}_{j}$ exhibits a penalization value greater zero. We term this property $p$, where we left out the superscript $t$ for better readability.
    
    Together, properties \textit{s, r} and \textit{p} are equivalent with $\mathbf{x}^t$ encompassing a solution where each qubit $i\in V$ is mapped to exactly one core $j$, where less or equal to $c_j$ qubits are assigned to a core $j$ and where each adjacent pair of qubits in the graph $G_t(V,E_t)$ is assigned to the same core. The said properties require $S,R$, $P$ to be zero implying the same for $F$. Thus if $F$ is zero the properties \textit{s, r} and \textit{p} hold.
\end{proof}
The final step for the qubit assignment part is to generalize the objective function for the qubit assignment $F(\mathbf{x}^t)$ in Eq. (\ref{eq:F}) to all circuit slices $T$, i.e., $\mathbf{x}:= [\mathbf{x}^1, \dots, \mathbf{x}^{T}]$, which we term $H_a$
{\footnotesize
    \begin{align}
    \min_{\mathbf{x}} H_a (\mathbf{x}) 
     = \min_{\mathbf{x}} \sum^{T}_{t=1} F(\mathbf{x}^t) \label{eq:Ha}
    \end{align}}%
resulting in a solution to the qubit \textit{assignment} problem for all circuit slices.
\begin{proof}
    Since $H_a$ only consists of functions $F$ with a minimum value of zero, $H_a$ is zero iff $F$ is equal to zero for all circuit slices $t$. This is equivalent to the property that \textbf{x} presents a valid assignment for all circuit slices.
\end{proof}

\subsubsection*{b) Minimizing inter-core communication}
As discussed in Ex.~\ref{ex:mapping}, the other task of the objective function is to minimize potential inter-core communication between every pair of assignments $(F_{t-1}, F_t)$. 

Let us consider an arbitrary qubit $i$ in two subsequent times slices: $x_{ij}^{t-1}= x_{il}^{t}=1$. If $j\ne l$, then $i$ is assigned to core $j$ in circuit slice $t-1$ and to a different core $l$ in circuit slice $t$, and thus a state transfer is necessary between cores $(j,l)$. Hence, we can define the following penalization objective function for inter-core communication:
{\footnotesize
    \begin{align}
    \min_{\mathbf{x}} H_t (\mathbf{x}) 
     = \min_{\mathbf{x}} \sum^{T}_{t=2} d(\mathbf{x}^{t-1},\mathbf{x}^{t})
     \label{eq:Hm}
    \end{align}
where
    \begin{align}
     d(\mathbf{x}^{t-1}, \mathbf{x}^{t}):= \sum^n_{i=1}d_{jl} x^{t-1}_{ij}x^{t}_{il}\label{eq:d}
    \end{align} }%
where $d_{jl}$ is the hop count between cores $(j,l)$, i.e. the number of links a state traverses in order to travel from core $j$ to $l$. E.g., in the simplest case, each core is connected to all other cores (all-to-all) via a link over which quantum state teleportation can take place, then  $d_{jl}=1\ \forall(j,l)$ where $i\ne j$ and trivially $d_{jj}=0 $. In order to count the number of inter-core communications $M$, we can readily use Eq. (\ref{eq:Hm}):
\begin{example}[Count inter-core communications] \label{ex:countstate transfers}
    Given the assignments of Ex. \ref{ex:mapping}, qubit 4 is the only qubit that requires a state transfer between circuit slices 1-2. Using Eq. (\ref{eq:d}) yields
    {\footnotesize
    \begin{align*}
        d(\mathbf{x}^{1}, \mathbf{x}^{2})&=\underbrace{d_{11}}_0 x^1_{01}x^2_{01} + \underbrace{d_{11}}_0 x^1_{11}x^2_{11}  + 
        \underbrace{d_{22}}_0 x^1_{22}x^2_{22} 
        \\ &+ \underbrace{d_{11}}_0 x^1_{31}x^2_{31} + \underbrace{d_{12}}_1 x^1_{41}x^2_{42} +
        \underbrace{d_{22}}_0 x^1_{52}x^2_{52}\\&= 1
    \end{align*}}
    as expected.
\end{example}

Combining objectives (\ref{eq:Ha}) and (\ref{eq:Hm}) gives the final form of the quadratic objective function. The QUBO model for quantum circuit mapping to multi-core quantum devices thus reads
{\footnotesize
    \begin{align}
    \min_{\mathbf{x}} H (\mathbf{x}, \lambda)
     = \min_{\mathbf{x}} \left[H_a(\mathbf{x}) + \lambda\cdot H_t(\mathbf{x}) \right]\label{eq:H}
    \end{align}}
where $\lambda\ge0$ is a weighting parameter to scale $H_t$.

\begin{remark}
The proof of $H_a$ shows that if $H_a > 0$ at least one assignment is not valid and the circuit cannot be executed. Let us assume, for example, for a solution $\mathbf{x}_1$, the objective function results in $H_a=1$ and $H_t=1$. For another solution $\mathbf{x}_2$, the sums result in $H_a=0$ and $H_t=3$. In total the objective functions for $\mathbf{x}_1$ and $\mathbf{x}_2$ yield $2$ and $3$, respectively. Even though the solution $\mathbf{x}_1$ yields a wrong assignment, i.e., $H_a=1$ (with one inter-core communication $H_t=1$) and  $\mathbf{x}_2$ yields a correct assignment (with 3 inter-core communications $H_t=3$), the total objective function would decide in the favor of $\mathbf{x}_1$, which is undesirable.
\end{remark}
We therefore choose $\lambda$, such that $0 \le \lambda\cdot H_t \lesssim 1.$
with 
{\footnotesize
 \begin{align}
     \lambda \lesssim (Tn )^{-1}\label{eq:lam}
 \end{align}}
hence
{\footnotesize
 \begin{align*}
    \underset{\mathbf{x}^t}{\text{max}}\ H_t = \underset{\mathbf{x}^t}{\text{max}} \sum^{T}_{t=2} d(\mathbf{x}^{t-1},\mathbf{x}^{t}) < T \cdot n
 \end{align*}}%
in the all-to-all case, which favors correct assignments with a higher likelihood. Further, we calculate $H_a(\mathbf{x})$ in order to verify the correctness of $\mathbf{x}$, i.e., we only consider a solution to be valid, if $H_a(\mathbf{x})=0$.

\FloatBarrier
\subsection{Toy-model example}
In order to gain an intuitive understanding of how the weighting parameter $\lambda$ impacts the mapping solution, we show the outcome of the mapping of a simple circuit (Fig. \ref{fig:toyexample}) to a three-core quantum system, each core with a capacity of two physical qubits. We assume the cores are all-to-all connected, i.e., $d_{ij}=1\ \forall ij$ where $i\ne j$ and trivially $d_{jj}=0$. The overview of the given parameters is listed in Tab. \ref{tab:toyexample}.\\ 

\begin{table}[!ht]
    \centering
\resizebox{0.7\linewidth}{!}{
    \begin{tabular}{r|l}
        Parameter & Value \\ \hline
        $n$ (\# qubits)& $5$ \\
        $T$ (\# circuit slices)& $5$ \\
        $\lambda$ (penalization parameter) & $0.001, 0.1$ \\
        $k$ (\# cores) & 3 \\
        $c_j = c \ \forall j$ (capacity) & 2 \\
        \vspace{0.2cm}
    \end{tabular}}
    \caption{Parameter overview of toy-model example}
    \label{tab:toyexample}
\end{table}

\begin{figure}
    \centering
    \includegraphics[width=0.6\linewidth]{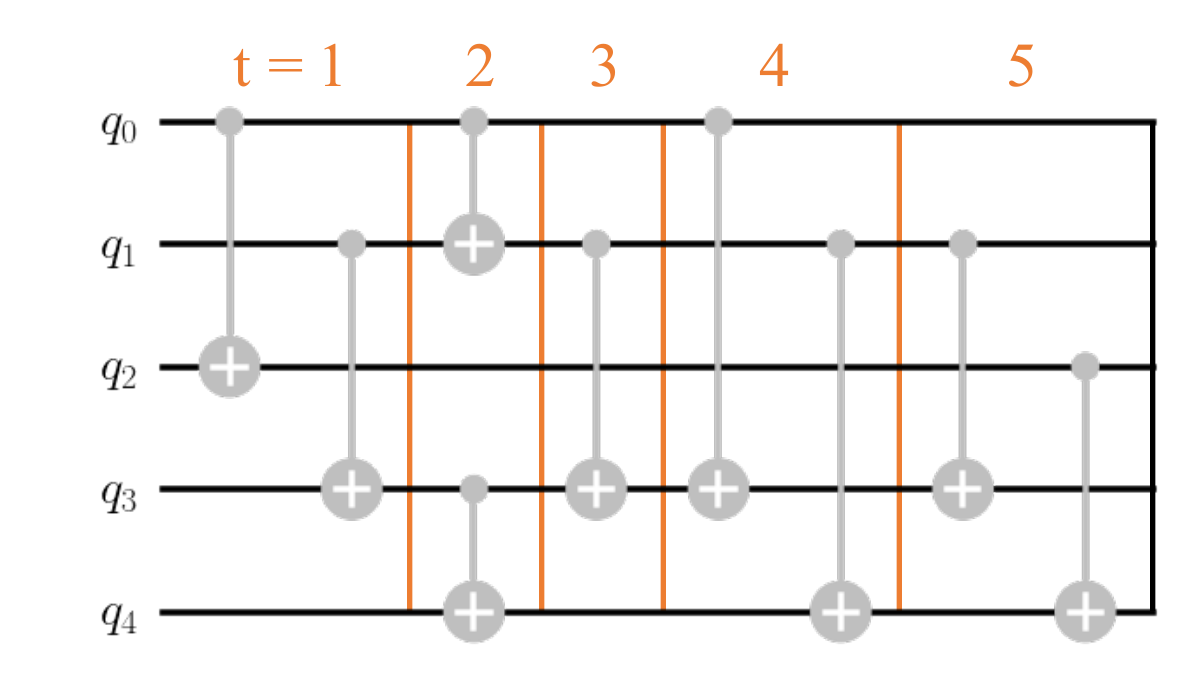}
    \caption{Toy-model quantum circuit consisting of nine two-qubit gates in five circuit slices $t\in[1,5]$ labeled in orange.}
    \label{fig:toyexample}
\end{figure}

We can determine the size of the solution array \\$N = T\cdot k \cdot n = 5 \cdot 3 \cdot 5 = 75$, i.e., $\mathbf{x} = \{0,1\}^{75}$. Further, the problem requires  $T \cdot c \cdot k = 5 \cdot 2 \cdot 3 = 30$ slack variables, i.e., $\mathbf{y} = \{0,1\}^{30}$, yielding $105$ variables in total. We can assemble the graph Laplacians. Considering the interaction graph $G_1(V,E_1)$, $L^1$ reads
{\footnotesize
$$ L^1 = \begin{pmatrix}
    1 & 0 & -1 & 0 &  \\
    0 & 1 & 0 & -1 & 0 \\
    -1 & 0 & 1 & 0 & 0 \\
    0 & -1 & 0 & 1 & 0 \\
    0 & 0 & 0 & 0 & 0 \\
\end{pmatrix}.
$$}%
It is a simple exercise to determine the remaining Graph Laplacians, which is left to the reader. Given all the constants we can assemble the QUBO problem (see Eq. \ref{eq:H}), which we then solve with a simulated annealing heuristic for two different $\lambda$ settings, $\lambda_1=0.001$ and $\lambda_2=0.1$.\\

Both mapping solutions are displayed in Figs. \ref{fig:sol1} and \ref{fig:sol2} are valid mappings, since all two qubits are placed such that they can conduct their two-qubit gates (Fig. \ref{fig:toyexample}). The mapping with $\lambda_1$ resulted in 12 inter-core communications, while the model using $\lambda_2$ resulted in 8 inter-core communications.  As expected, but still significant, a higher $\lambda$ parameter yielded $33\%$ improvement in the number of inter-core communications. Note that the weighting of $\lambda$ is restricted by the validity of an assignment (Eq. \ref{eq:lam}); meaning, even though the optimal (highest) $\lambda$, which still gives a valid assignment ($H_a = 0$), is \textit{not} a priori known, we do know that a weight set too high can jeopardize the feasibility of the solution, which is why the method must verify $H_a = 0$ to guarantee a valid assignment.

\begin{figure}
    \centering
    \includegraphics[width=0.7\linewidth]{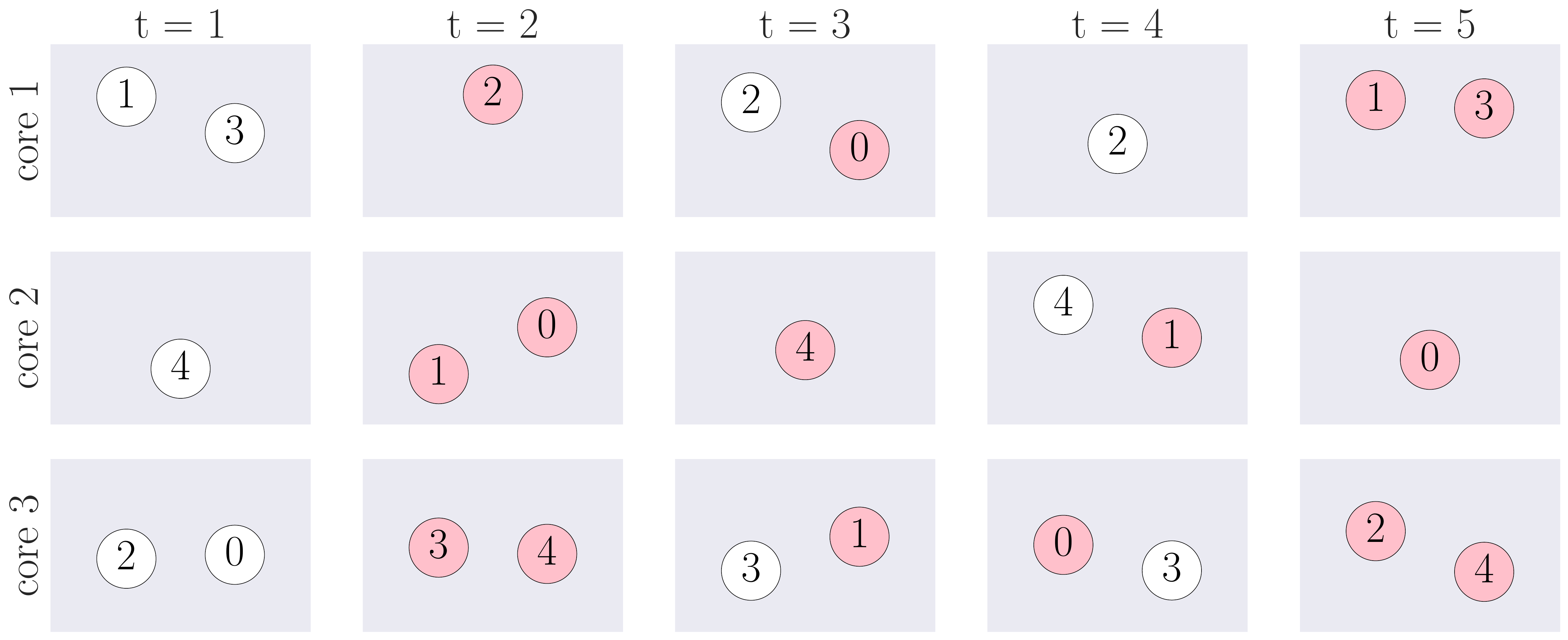}
    \caption{$\lambda_1 = 0.001$ mapping solution of the QUBO model, where the x-axis is the timeline and the y-axis marks the different cores. Red-colored nodes signify that these qubit states have been transferred from the previous slice. This mapping hence resulted in 15 inter-core communications.}
    \label{fig:sol1}
\end{figure}

\begin{figure}
    \centering
    \includegraphics[width=0.7\linewidth]{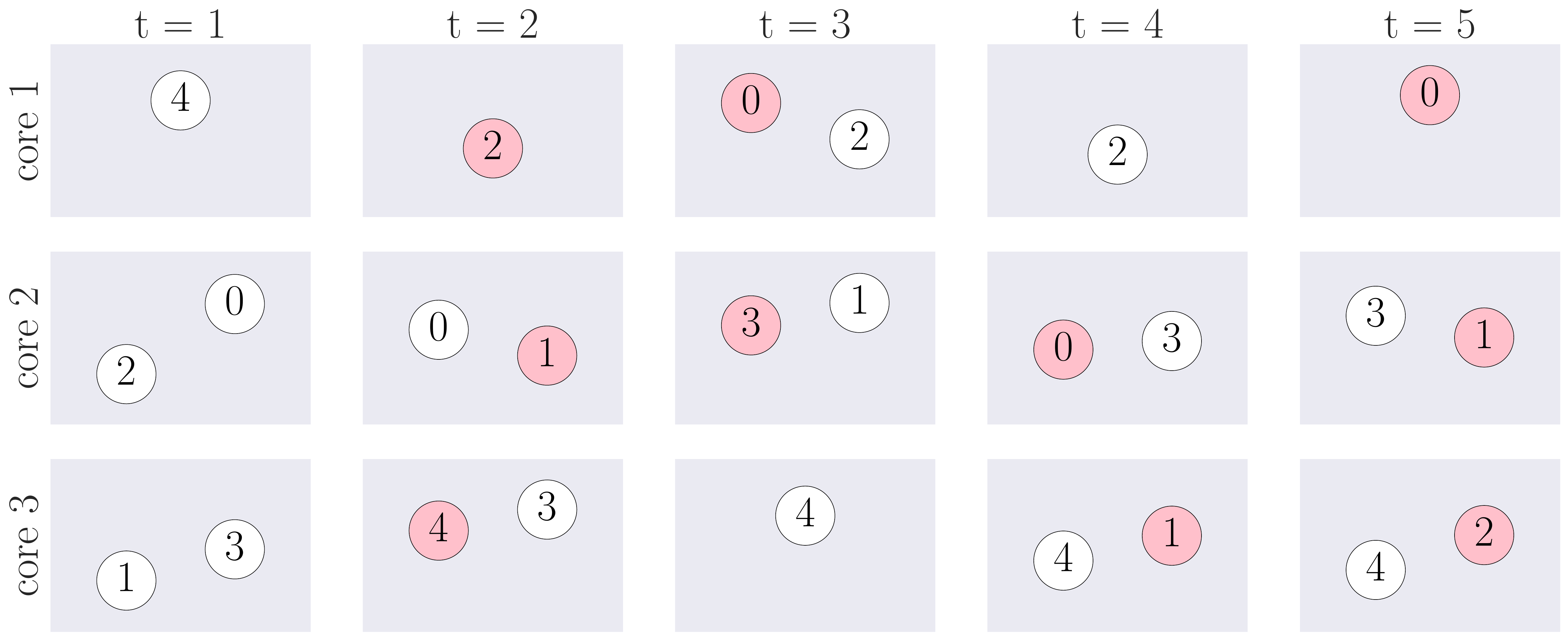}
    \caption{$\lambda_2 = 0.1$  mapping solution of the QUBO model, where the x-axis is the timeline and the y-axis marks the different cores. Red-colored nodes signify that these qubit states have been transferred from the previous slice. This mapping hence resulted in 10 inter-core communications.}
    \label{fig:sol2}
\end{figure}

\subsection{Time Slices of a Quantum Circuit}  \label{Slicing}
Solving the quantum circuit mapping problem for multi-core systems with the introduced formulation requires a quantum circuit to be partitioned into smaller blocks of gate sequences, where each so-called circuit slice $t$ is represented by an interaction graph $G_t(V,E_t)$. The objective is to generate graphs, where each can represent as many gates of the circuit as possible, provided that all gates in a slice can be executed without moving qubits between cores. To do so, we use a recursive procedure, outlined in Alg. \ref{alg:slicing}.

Given a list of gates $\Psi$ of a quantum circuit and a list of empty graph objects stored in the list \textit{slices}, we iterate over all gates and apply the function \textsc{addGate} in each step. The graph objects are initially empty and are filled over the course of the iterations, with every iteration variable $t$ is updated accordingly, marking the number of non-empty graph objects in \textit{slices}, i.e., the number of graph objects, which contain at least one edge. \textsc{addGate} then decides to which graph in the list the gate is assigned. Let us consider some arbitrary gate $\Psi_l =(i_1,i_2)$ involving qubits $i_1,i_2 \in V$, one of the following four cases is about to happen: 
\begin{itemize}
    \item[(1)] the function returns immediately if the gate is already present in the current slice $t$;
    \item[(2)] the gate needs to be assigned to the subsequent slice $t+1$ if either or both of the qubits $i_1,i_2$ are involved in the current slice in least one other gate. In other words a gate ($i_1,i_2$) cannot be added to a slice, if another gate ($i_1,.$) or ($i_2,.$) is already present, as then the gates are not concurrently executable; 
    \item[(3)] the gate is simply added to current circuit slice $t$, if it is the first slice $t=0$ and 1) and 2) are False, i.e., if the gate is not present and none of the qubits has been used;
    \item[(4)] the function is called again using the previous slice $t-1$, if the current slice is not the first, i.e., $t>0$ and 1) and 2) are False.
\end{itemize}
\renewcommand{\algorithmicrequire}{\textbf{global }}
\begin{algorithm}
\begin{minipage}{0.95\linewidth}
        \centering
\scriptsize
\textbf{Input} $\Psi$, slices\Comment{ list of quantum circuit gates, list of empty graphs}\\\vspace{.1cm} 
\textbf{Output} slices \Comment{ list of interaction graphs $G_t$}\\\vspace{.1cm}
\begin{algorithmic} 
\scriptsize
    \State \Function{addGate}{quantum circuit gate, $t$} \Comment{defines cases (1)-(4)} 
    \State \uIf{gate in slices[t]} {return;}\uElseIf{any(gate$_i$ used in slices[t])}{
    \State $t\gets t+1$
    \State slices[$t$] $\gets$ gate
    \State return;}\uElseIf{$t$ is 0}{
    \State slices[$t$]$\gets$ gate 
    \State return;}\Else{return \textsc{addGate}(gate, $t-1$);}
\EndFunction\vspace{.2cm} 
\State \algorithmicrequire slices  \Comment{list of empty graph objects}\vspace{.2cm}  
\State  \For{gate in $\Psi$}{\algorithmiccomment{main iteration, conducts slicing}
\State $t \gets$ N.o. non-empty graphs in slices 
\State \textsc{addGate}(gate,\ $t$)
}
\end{algorithmic}
    \caption{Slicing the quantum circuit}
    \label{alg:slicing}
\end{minipage}
\end{algorithm}

\section{Experimental setup} \label{Sec4}

All tests were run on an Intel(R) Xeon(R) Silver 4210 CPU @ 2.20GHz with 256 GB of RAM and 16 cores running Debian GNU/Linux 11. The method is implemented with Python 3.9.2 using Qiskit's framework \cite{Qiskit} to generate the benchmark programs and to (pre-)process the circuits. We used dwave-neal version 0.5.9 \cite{dwaveweb} simulated annealer \cite{kirkpatrick1983optimization} as our dedicated solver heuristic, as it comes with the open source package of D-Wave Systems and can compete with commercial QUBO solvers \cite{oshiyama2022benchmark}. It is worth noting that the selected solver can be easily replaced with any other solving method for quadratic programming. Due to the memory allocation limits of our system, we implemented a divide-and-conquer approach to divide and solve more amenable subsets of the problem before adding them back together.  
Considering the largest 50\% of our problem sizes, up to four divisions are necessary (where one division was sufficient for 62\% of them). Divide-and-conquer allows us to solve arbitrarily large QUBO instances, with the trade-off of adding some locality to the search process. Since 50,000 variables span about 80 time slices on average in our benchmarks, the added locality is still far from what is typically used for a local estimate.

\subsection{Benchmarks}

Similar to previous works \cite{Rodrigo2021, baker2020time}, we use as benchmarks: 

\begin{itemize}
    \item Subroutines which are building blocks of larger circuits (e.g., Shor's algorithm) such as the Quantum Fourier Transform (QFT), the Multi-Target Gate, Draper's QFT Adder \cite{draper2000addition} and Cuccaro's Ripple-Carry Adder \cite{cuccaro2004new} (both in their \textit{fixed} version) utilized as Qiskit circuit implementations\footnote{\url{https://qiskit.org/documentation/apidoc/circuit_library.html##module-qiskit.circuit.library}, 18.04.2023}.

    \item Four instances of randomly generated circuits to fill the gap in structured circuits in terms of the following parameters: i) the number of gates and qubits; ii) circuit depth; and iii) circuit density, that is the average number of two-qubit gates divided by the circuit depth. In this last category, we also include Quantum Volume \cite{cross2019validating}, which is a circuit having equal values for width and depth used to evaluate the overall performance of a quantum computer. Note that quantum volume circuits are by definition challenging to map due to their high gate density.\cite{cross2019validating}.
\end{itemize}
We generated a total of 407 benchmark instances for these nine circuits, which parameters are shown in Fig. \ref{fig:benchmarkoverview}. Note that they cover qubit counts ranging from 50 to 100 qubits, circuit depths from 13 to more than 1,100, and gate densities ranging from 0.67 to 21.42. Our QUBO-based mapping framework as well as benchmark instances are available in an open-source format and can be accessed via  \href{https://github.com/Luisenden/map-quantum-circuits-to-multi-core}{github}. We employed Qiskit-based routines in order to ensure reproducibility by the scientific community.

It is important to note that we observed rather large differences between the circuit instances we used and the ones in \cite{baker2020time}. For instance, the QFT adder of the same number of qubits in \cite{baker2020time} has significantly lower circuit depth, and random circuits of the same number of qubits and depth have a much higher number of gates. Therefore, our benchmark instances are similar to the benchmarks from \cite{baker2020time} in terms of functionality and structure, but the results are not directly comparable.

\begin{figure}[ht]
\centering
\includegraphics[width=0.8\linewidth]{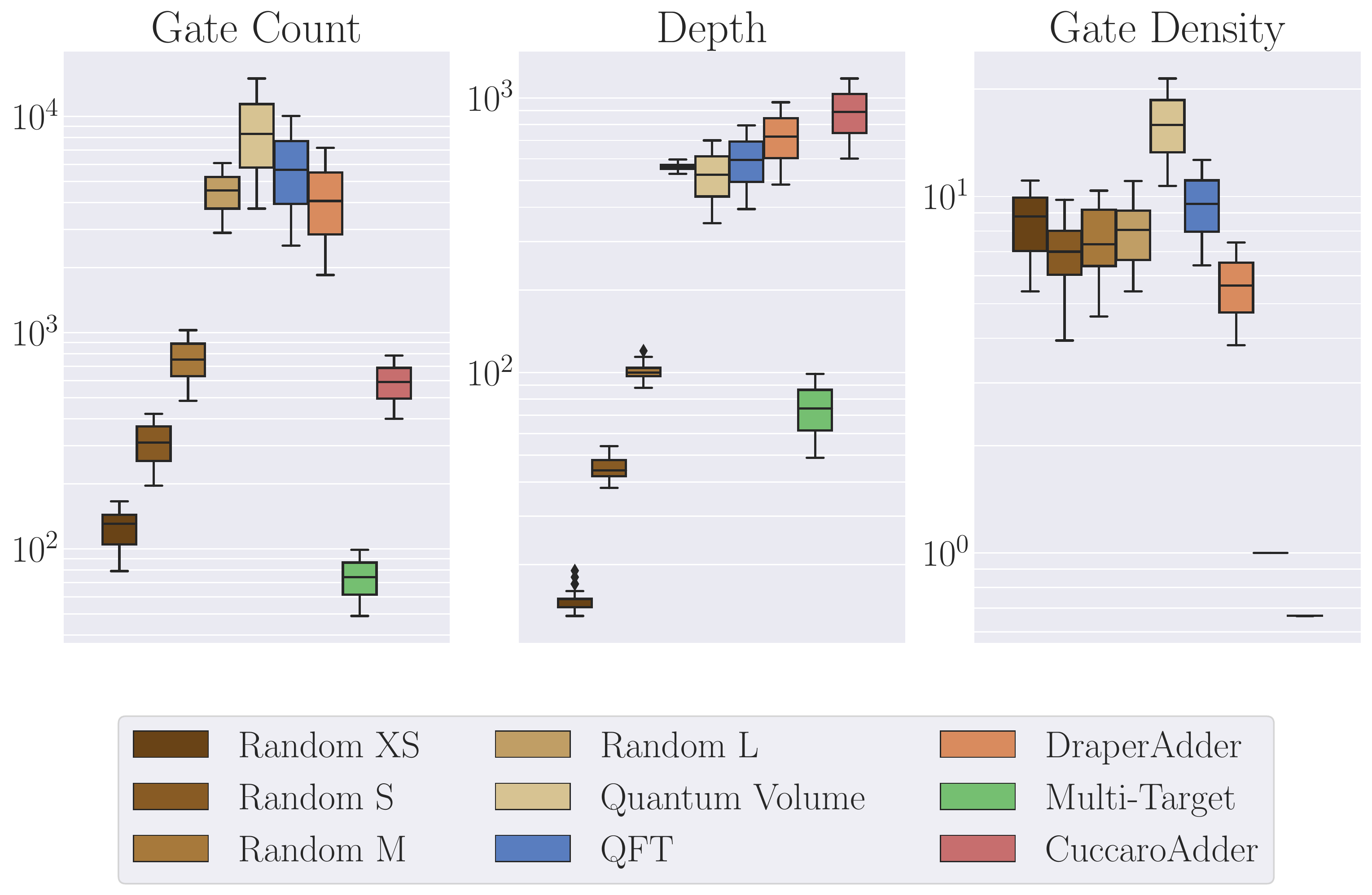}
\caption{Range of relevant circuit parameters of our benchmark set. Circuits between 50 to 100 qubits have been generated using Qiskit's circuit library and decomposed into CNOT gates. Random XS, S, M and L refer to the different depth intervals 13-19, 38-54, 88-120, 529-596, respectively. Note that all random circuits have an average circuit density of $\sim 10$ two-qubit gates relative to depth, whereas Multi-Target Gate and Cuccaro Adder show a much lower density, 1 and 0.5 respectively.}
\label{fig:benchmarkoverview}
\end{figure}

\subsection{Multi-core quantum computing architecture}
We consider two multi-core architecture topologies comprising 10 cores, each of them consisting of 10 all-to-all connected qubits. We assume all-to-all qubit connectivity within the cores since intra-core communication is negligible compared to inter-core communication. The first topology (Fig. \ref{fig:architectures}a) is similar to the one considered in \cite{baker2020time}, with all-to-all core connectivity. The second topology (Fig. \ref{fig:architectures}b) showcases more realistic inter-core connectivity for future developments in modular quantum computing as shown in \cite{laracuente2023modeling}. In order to have the same number of cores in both topologies, we choose a $2 \times 5$ grid layout. Our method is, however, easily adaptable to include more complex architectures as well, which would include intra-core qubit routing like already existing modular architectures shown in \cite{laracuente2023modeling}. In addition, note that our method is not restricted to any particular size of the architecture.

\begin{figure}[t!]
	\centering
	\includegraphics[width=0.7\linewidth]{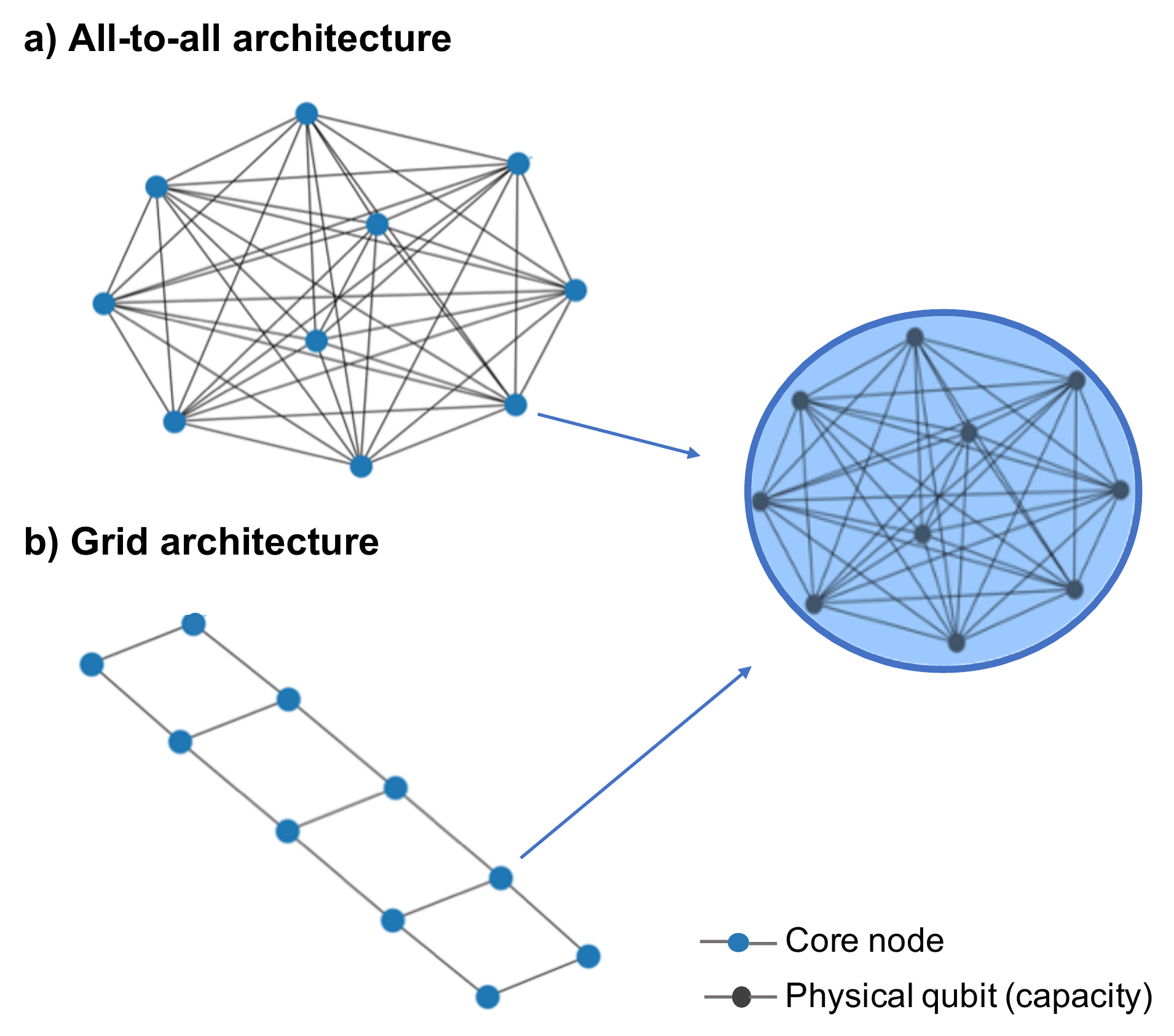}
 	\caption{Two different multi-core architectures considered in our experiments: \textbf{a)} All-to-all connected cores and \textbf{b)} 2D Grid core connectivity. Each node in the two graphs represents a core and the edges correspond to communication links between the cores. On the right, the intra-core qubit topology is shown, consisting of 10 all-to-all connected qubits.}
 	\label{fig:architectures}
\end{figure}

\subsection{Performance metrics}
Our main performance metric is the number of inter-core communications $M$ between cores, also termed non-local communications, required for executing a circuit on the quantum multi-core system. They are calculated using $H_t$ as previously shown in Ex. \ref{ex:countstate transfers}. In addition, we provide the execution time for all mapping attempts.

\section{Results and discussion}
\label{Sec5}

\subsection{Number of inter-core communications}
In this section, we analyze and discuss the results of our experiments, in which we tested our method with the circuits as described above on the two different multi-core architectures depicted in Fig. \ref{fig:architectures}. As we just mentioned, we use as performance metrics  the number of inter-core communications and execution time and analyze  how they relate to the structure of the circuits, in particular circuit density, and how they scale with a larger number of qubits, gates and depth.
\subsubsection{Layout all-to-all}

Fig. \ref{fig:results} depicts the results of all benchmarks in this study. We ran 50-100 qubit-sized circuits for each benchmark circuit (except DraperAdder and CuccaroAdder, which exhibit only even or odd numbers of qubits in their fixed version). We obtained a success rate of 87\%, which means a valid assignment could be found in the first attempt for 354 out of 407 circuits.

Fig. \ref{fig:results}.\textbf{a)} shows the number of inter-core communications for all benchmark circuits. The best performers are clearly Random XS, Random S, Random M and Multi-Target, which stayed below 5,000 inter-core communications. This is not surprising, as they cover the lower range of gate count and depth (Fig. \ref{fig:benchmarkoverview}), i.e., smaller circuits require smaller numbers of inter-core communications. On the higher end of the performance measure lie the adders, QFT, Random L and Quantum Volume, which yielded inter-core communications between 2,000 and 17,000. The trend of circuits with more gates and higher depth resulting in higher inter-core qubit communications is quite clear. However, the QFT and Quantum Volume benchmarks displayed an opposite tendency, as the Quantum Volume benchmark resulted in 2,800 less inter-core communication on average compared to QFT, even though Quantum Volume has about $2,000$ more two-qubit gates, $1.7$ times higher density and similar depth. 

When we look at the inter-core qubit communications relative to the number of two-qubit gates in a circuit \ref{fig:results}.\textbf{b)}, the ranking of benchmarks changes significantly. In this picture, the benchmark Multi-Target, though presenting one of the smaller absolute outcomes, yields up to five inter-core communications per two-qubit gate, similar to our largest benchmark, the CuccaroAdder. The structure of these circuits can serve as an explanation, as both of them are hardly parallelizable: a circuit layer comprises only at most one two-qubit gate. This circuit structure implies circuit slices where $|E_t|$ is small and qubit assignments between circuit slices are almost identical. This is a disadvantage for the proposed method, as it performs better when a significant number of changes is required between circuit slices(which is typically the case for dense, i.e., highly parallelized circuits): considering the 50\% of the circuits with the highest density, the relative inter-core communications stay below 0.8 on average (less than one communication per two-qubit gate), while the other half of the circuits exhibit 19 inter-core communications per gate on average. The Quantum Volume benchmark accentuates these findings further: with 8,000 two-qubit gates it has by far the highest gate count and highest parallelization, and it is the second-best performer in the relative picture \ref{fig:results}.\textbf{b)} with one inter-core communication every second two-qubit gate.  \ref{fig:results}.\textbf{c)}. depicts the weighting factors $\lambda$ (\ref{eq:lam}), chosen between $0.7$ and $0.006$. As the term $H_t$ increases proportionally to the circuit size, we choose $\lambda$ accordingly to achieve $\lambda H_t \lesssim 1$, we can see this trend clearly. It can be observed, that on the selected value, the number of inter-core communications increases, which is the expected behavior. For the sake of a high success rate and efficiency in producing the first mapping results, we determined a conservative selection of these weighting parameters.

\FloatBarrier
\begin{figure}[!ht]
\includegraphics[width=\linewidth]{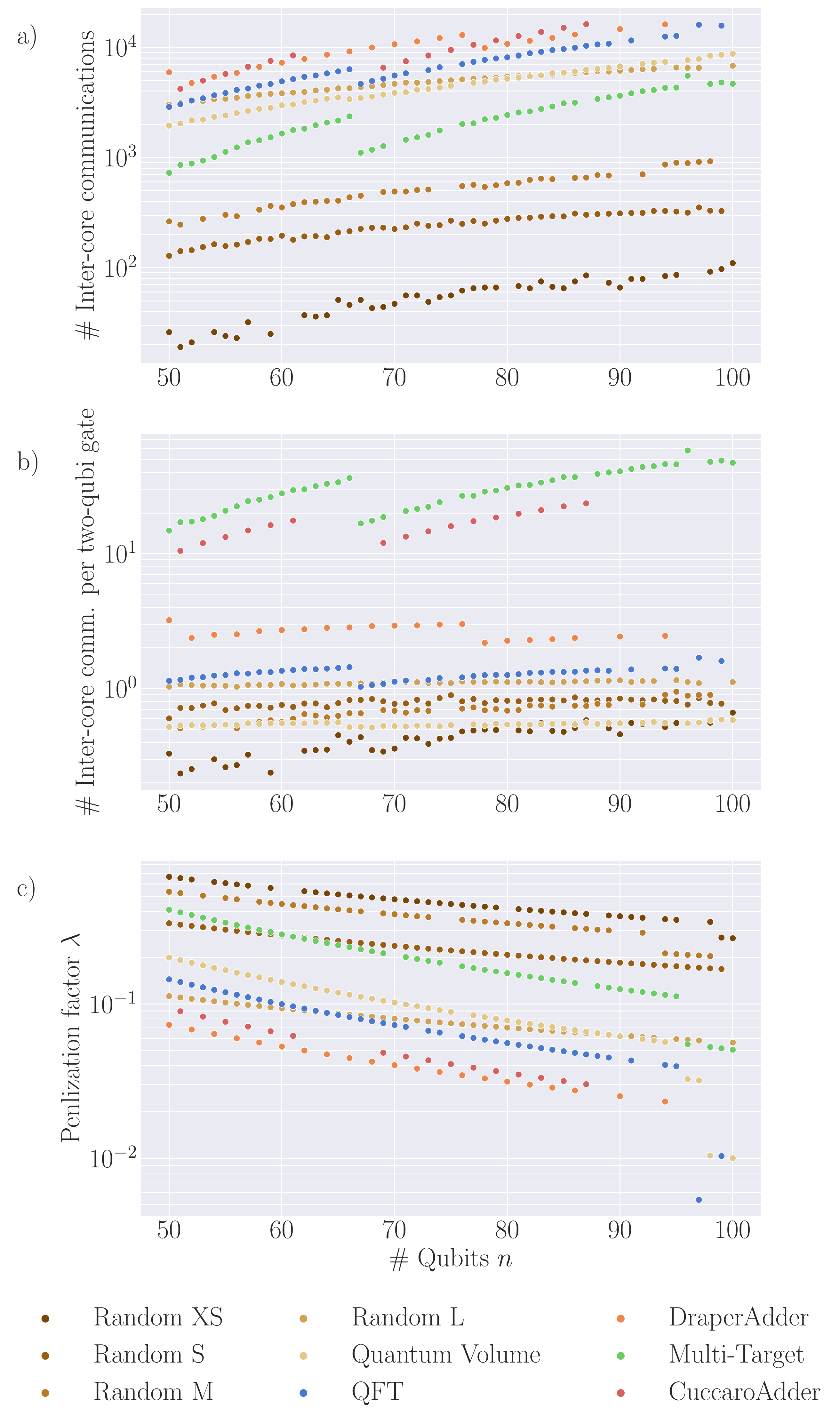}
\caption{Results for all benchmarks run on an all-to-all multi-core layout. \textbf{a)} and \textbf{b)} summarize the inter-core communications in absolute and relative values to two-qubit gates, respectively, for successful mappings between 50 and 100 qubits. The tendency for higher amounts of inter-core communication for larger circuits is clearly detectable. However, from a relative to two-qubit gate point of view, the circuits with high density yield lower outcomes, i.e., random circuits and QFT. \textbf{c)} depicts the penalization factor $\lambda$ chosen for each run.}
\label{fig:results}
\end{figure}

\subsubsection{2D-Grid layout}
In this experiment we tested the method for qubit sizes $n\in[50,75,100]$ on a $2\times 5$ grid quantum architecture layout as depicted in Fig. \ref{fig:architectures}\textbf{b)}. Tab. \ref{tab:2dgrid} lists the results, showing a similar dependency on the problem size in terms of inter-core communications of the circuit to the all-to-all layout experiment. The number of inter-core communications is however higher than the results of the first experiment due to longer hop distances between cores; for all 50-qubit circuits, for example, the number of inter-core communications is about twice as high compared to the all-to-all layout. With both layouts, we can detect better performance values for random circuits. Our selected quantum circuits were characterized in terms of standard parameters, such as the number of qubits and depth. However, there are other parameters, such as the percentage of two-qubit gates, and interaction graph characteristics, which were shown to have correlations to what can be expected from a mapping outcome \cite{bandic2022interaction}. Since random and real circuits can differ significantly in their inherent circuit structure, the latter is expected to have an impact on the performance as well.
\begin{table}[ht]
    \centering
    \resizebox{0.9\linewidth}{!}{%
    \begin{tabular}{r|rrrr||r}
    Benchmark &  \# Qubits &  Depth &  \makecell{Two-qubit \\gate count} &     $\lambda$ &  \makecell{\# Number of \\ inter-core comms. }\\
    \hline
 &       50 &    350 &      3750 &  0.1020 &    3036 \\
Quantum Volume &       75 &    150 &      8325 &  0.0450 &    7646 \\
 &      100 &    700 &     15000 &  0.0100 &   17940 \\\hline
   &       50 &     49 &        49 &  0.2041 &    1261 \\
Multi-Target   &       75 &     74 &        74 &  0.0541 &    5097 \\
   &      100 &     99 &        99 &  0.0101 &   13218 \\\hline
      &       50 &     28 &       171 &  0.2333 &     182 \\
Random S       &       75 &     26 &       253 &  0.2222 &     273 \\
      &      100 &     30 &       344 &  0.1000 &     540 \\\hline
       &       50 &     65 &       412 &  0.1333 &     723 \\
Random M       &       75 &     61 &       647 &  0.1778 &     923 \\
     &      100 &     65 &       867 &  0.0667 &    1888 \\\hline
     &       50 &    350 &      2474 &  0.0225 &    7308 \\
Random L       &       75 &    357 &      3730 &  0.0150 &   11789 \\
      &      100 &    344 &      5004 &  0.0112 &   15778 \\\hline
  &       49 &    577 &       384 &  0.0389 &    7450 \\
CuccaroAdder  &       75 &    889 &       592 &  0.0163 &   19827 \\
  &       99 &   1177 &       784 &  0.0046 &   45275 \\\hline
         &       50 &    394 &      2525 &  0.0816 &    3947 \\
QFT            &       75 &    594 &      5661 &  0.0450 &    7443 \\
          &      100 &    794 &     10050 &  0.0025 &   36800 \\\hline
    &       50 &    484 &      1850 &  0.0610 &    7769 \\
    DraperAdder &       76 &    744 &      4294 &  0.0287 & 16675 \\
     &      100 &    984 &      7450 &  0.0017 &   54384 \\\hline \vspace{0.1cm}
    \end{tabular}}
    \caption{Results for selection of benchmarks on a $2\times5$-grid architecture.}
    \label{tab:2dgrid}
\end{table}

\subsection{Execution Time}
Since execution time depends on the problem size \\$N= T\cdot n \cdot k$ (i.e., the circuit size and multi-core system size), but not on the factors in the core-distance matrix $d$, distinguishing between the two architecture-layout experiments is not required. Fig. \ref{fig:runtime} shows the execution runtimes which resulted in values between 0.1 and 138.36 minutes, where 3/4 of the circuits were successfully compiled in 30 minutes or less. Although these execution times are high for a compiler method, we must highlight that this method is classical by nature. One can therefore justify usage of our compilation strategy, as savings of even a few milliseconds are consequential in the success probability of a quantum algorithm \cite{siraichi2018qubit}, especially in the case of inter-core communication \cite{humphreys2018deterministic}. Furthermore, for ease of reproducibility and comparability, there were no HPC resources or similar performance-improving tools employed, which are typically accessible for commercial use.

\begin{figure}
    \centering
    \includegraphics[width=\linewidth]{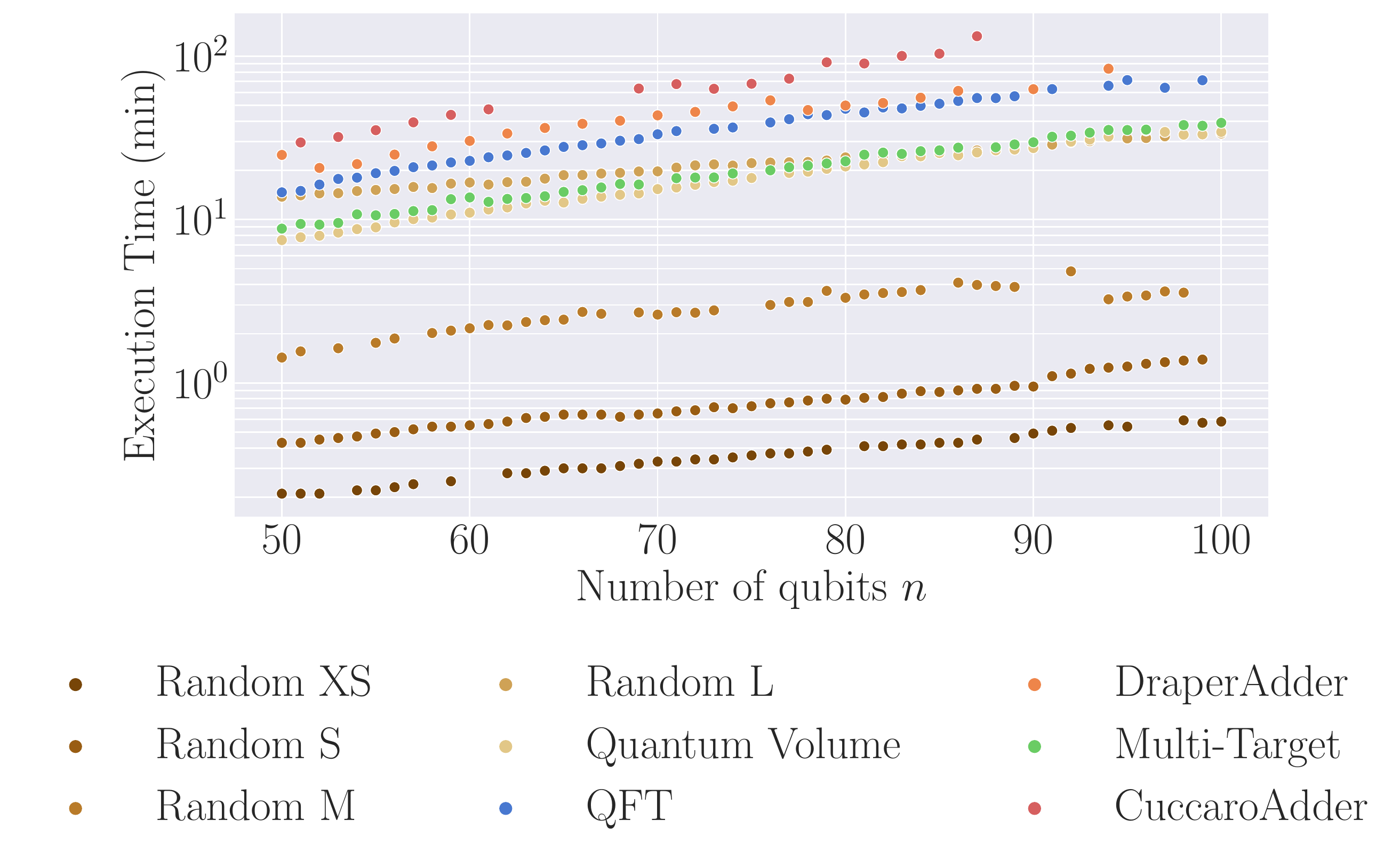}
    \caption{Execution times of the QUBO method, which resulted in a valid mapping solution. The results span a range from the smallest circuit (50-qubits Random XS) of about 12 seconds up to over 138 min for the largest instance (100-qubits CuccaroAdder).}
    \label{fig:runtime}
\end{figure}

\section{Conclusion and future work}

\label{Sec6}

Multi-core quantum computing architectures are one of the most promising approaches towards large-scale quantum computers, for which it is required to develop new quantum circuit mapping techniques that consider the inter-core communication requirements.

In this study, we proposed a novel approach for mapping quantum circuits to multi-core quantum architectures based on QUBO. The main strengths of our method lie in the formulation of the QUBO itself, as the structure permits i) the decoupling of the problem definition from the solver, as well as ii) superseding limitations of look-ahead approaches used in previous solutions. We tested the method's functionality for a wide range of  benchmarks on two different  multi-core architecture layouts composed of 10 cores with a capacity of 10 qubits each. 
Taking stock of the analysis of our benchmark experiments, we expect a success rate of 87\% to find a solution, where the most promising results could be achieved with circuits exhibiting high density and  shallow circuits \cite{blume2020volumetric}. 
Note that highly-parallelized circuits are usually the most challenging to run on quantum devices and are therefore often used as benchmarks to test them, which makes them most relevant to quantum computing of the near- to mid-term future \cite{broadbent2009parallelizing}. Furthermore, our method is easily adjustable to different architecture layouts as these changes require only altering the factors in the objective function. As a direct and positive consequence, these adjustments do not affect the execution time of the method. 

Our method, on the other hand, faces some challenges with the scalability of quantum circuits in terms of circuit depth. In order to improve our approach, including its scalability capabilities, in our future work we will be focused on:
\begin{itemize}
    \item Finding a suited weighting parameter $\lambda$, which is, in general, a non-trivial task and is a matter of ongoing research \cite{verma2022penalty}. Based on these methods we can employ an existing technique that fits our problem instances, which will be our first solution for improving the circuit success rate.
    \item  Differentiating between remote two-qubit gates (operations between two separate qubits in different cores) and qubit state transfers. This will help to better optimize  the amount of inter-core communication, by  doing the qubit transfers only when necessary, and also make more realistic multi-core architecture when combined with investigating other inter-core communication links \cite{10.1145/3579367,rodrigo2022characterizing}. 
    \item Solving the scalability limitations related to circuit depth by decreasing the number of decision variables via an alternative problem formulation to a gate assignment version instead of the qubit assignment. This alternative problem approach in its prototype state is already in the testing phase.
    \item Improving the runtime performance by parallelizing the algorithm execution (the algorithm is carried out sequentially at the moment) and by using the commercial QUBO problem solvers with superior computational power. The latter will also help to handle larger problem sizes (up to $10^6$ variables) \cite{oshiyama2022benchmark}.
\end{itemize}
Besides the described areas of improvement, an interesting direction for further experiments will be different layouts in both multi-core architecture and intra-core couplings, which will include topologies that are already proposed in \cite{laracuente2023modeling}. In addition to that we plan to perform an in-depth scalability analysis regarding the number of cores, physical and logical qubits as well as circuit depth, to find the actual limits of our algorithm as in \cite{ovide2023mapping}.

In summary, we expect that the introduced method exhibits potential for future qubit mapping tasks. It already demonstrates promising results for especially challenging mapping tasks to multi-core quantum architectures and exhibits the capability to be easily adjusted to both problem size and multi-core layout.




\section*{Acknowledgments}

MB and SF would like to acknowledge funding from Intel Corporation. EA and CGA acknowledge support from the EU, grant HORIZON-EIC-2022-PATHFINDEROPEN-01-101099697 (QUADRATURE). SA acknowledges support from the EU, grant HORIZON-ERC-2021-101042080 (WINC). 

\newpage

\bibliographystyle{IEEEtran}
\bibliography{References}

\end{document}